\documentclass{article}       
\usepackage{graphicx}
\usepackage{amsmath, amssymb, amsfonts}

\newcommand{\Rspace}        {\mathbb R}
\newcommand{\Zring}        {\mathbb Z}
\newcommand{\inter}        {\mathcal{I}}
\newcommand{\link}         {\mathcal{L}}

\newtheorem{Lemma}{Lemma}
\newtheorem{Theorem}{Theorem}
\newenvironment{proof}{Proof}

\begin{document}

\title{On the links of vertices in simplicial $d$-complexes embeddable in the euclidean $2d$-space\thanks{This work was done in part while the author was at
Fondation Sciences Math\'ematiques de Paris and D\'epartement d'Informatique, \'Ecole normale sup\'erieure, Paris, France. Part of this work was done while the author was at IST Austria.}}



\author{Salman Parsa}


\date{}

\maketitle

\begin{abstract}
 We consider $d$-dimensional simplicial complexes which can be PL embedded in the $2d$-dimensional euclidean space. In short, we show that
in any such complex, for any three vertices, the intersection of the link-complexes of the vertices is linklessly embeddable in the $(2d-1)$-dimensional euclidean space.
In addition, we use similar considerations on links of vertices to derive a new asymptotic upper bound on the total number of $d$-simplices in an (continuously) embeddable complex in $2d$-space with $n$ vertices, improving known upper bounds, for all $d \geq 2$. Moreover, we show that the same asymptotic bound also applies to the size of $d$-complexes linklessly embeddable in the $(2d+1)$-dimensional space.

\end{abstract}

\section{Introduction and Overview of Results} \label{section:intro}

There is great interest in understanding properties of simplicial complexes that are embeddable in a certain euclidean space. The basic case
is 1-dimen\-sional complexes that are embeddable in the plane, i.e., graphs that can be drawn on the plane without crossings between the edges. Planar
graphs are very well-understood. For instance, it is easily shown that if a planar graph has $n$ vertices it has at most $3n-6$ edges. However, roughly speaking, how ``dense'' a simplicial $d$-complex embeddable in $\Rspace^{2d}$ can be, is less understood for arbitrary $d$.
In this paper we show certain properties of embeddable complexes that, for instance, can be used to give an upper bound for this density problem. We note that
any simplicial $d$-complex can be embedded in a simplex-wise linear way in $\Rspace^{2d+1}$. However, for any $d \geq 1$, there exist simplicial $d$-complexes non-embeddable
in $\Rspace^{2d}$.

The property we prove in this article involves the notion of the link-complex of a vertex and linking of spheres in euclidean spaces (and more generally of algebraic cycles).
We begin by considering the simplest (but perhaps the hardest) case $d=2$. We consider first 2-complexes in $\Rspace^3$.
Let $K$ be a $2$-complex. The link-complex\footnote{This subcomplex is usually called the link of the vertex. In this paper, we call it the link-complex or the link-graph to prevent confusion with other usages of the word ``link".} of a vertex is the maximal $1$-subcomplex (a graph) whose join with the vertex is a subcomplex of $K$ (we give definitions in later sections).
Sometimes the embeddability of the complex provides restrictions on possible link-graphs. The following is well-known, see \cite{DeEd94}.
Assume that the complex $K$ is sitting in $\Rspace^3$, i.e., simplex-wise linearly embedded. Then, if we consider small enough balls around each vertex,
we can observe that the intersection of the boundary of a ball with the complex $K$ is a drawing of the link-graph of the vertex on the $2$-sphere.
Now a planar graph with $n$ vertices has at most $3n-6$ edges, hence the total number of edges in all link-graphs is at most $n(3(n-1)-6) = 3n^2 -9n$.
Since each triangle is counted three times it follows that such a complex $K$ over $n$ vertices has at most $n^2-3n$ triangles. A complex embedded in $\Rspace^3$ and with $\Omega(n^2)$ triangles
can be constructed by putting $n$ vertices on each of two skew lines in 3-space and then taking the Delaunay triangulation of the point set; it will consist of $\Omega(n^2)$ tetrahedra.
Alternatively, one can take the boundary of the $4$-dimensional cyclic polytope, remove a single facet and embed the result in $\Rspace^3$.

If we know that $K$ embeds in $\Rspace^4$, in general no restriction is imposed on the link-graph of a vertex. To see this, take an arbitrary graph in some $\Rspace^3 \subset \Rspace^4$
and ``cone'' this graph from a vertex on one side of the $3$-plane. Hence, arbitrary graphs appear as link-complexes of embeddable $2$-complexes. We can add another vertex on the other
side of the 3-plane and cone again. Thus the intersection of two link-graphs can be an arbitrary graph.
However, there are global restrictions on the set of all link-complexes and the above process cannot be continued. In brief, we show that in a PL embeddable $2$-complex in $\Rspace^4$, for any triple of distinct link-graphs, their common intersection (or a \textit{triple intersection}) is a linklessly embeddable graph. Informally, a linklessly embeddable graph is one that can be ``drawn'' in space without links between disjoint circles. Figure \ref{fig:linking} shows some
examples of links in euclidean spaces between spheres of the right dimensions.

\begin{figure}

\includegraphics[scale=0.8]{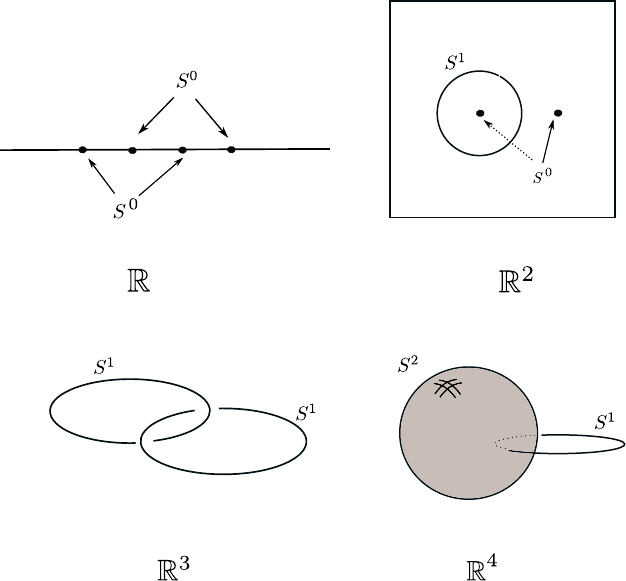}
\centering
\center{\caption{linking of spheres in various dimensions}}
\label{fig:linking}

\end{figure}

An interesting property of our main observation is that the same proof works for all dimensions. That is, we
show that for any PL embeddable $d$-complex, each triple intersection of link-complexes is a linklessly embeddable $(d-1)$-complex (in $\Rspace^{2d-1}$). More formally,

\begin{Theorem}\label{theorem:basic}
Let $\iota: |K| \rightarrow \Rspace^{2d}$ be a PL embedding of the $d$-complex $K$. If $L$ is a triple intersection of link-complexes of vertices of $K$, then $\iota_v : |L| \rightarrow S(v)$ embeds $L$ linklessly, for any $v$
that contains $L$ in its link-complex, where $\iota_v$ is the restriction to $|L|$ of the embedding of the (underlying space of the)  link-complex of $v$ into a small sphere $S(v)$ centered at $\iota(v)$.
\end{Theorem}

\paragraph{Remark}\footnote{Added after publication.}
If one is interested only in proving this theorem, one can indeed provide simpler proofs, for instance by using simple topological facts as used in \cite{Sk14,Sko03}. It is essentially an exercise in linking numbers. However, we deduce this theorem from a more general fact that proves equality of two linking numbers in opposite "link-complexes" up to sign, not just them being zero or non-zero, see Lemma 1. We chose this to get the most out of the proof technique. This needs translating intersections inside a 3-sphere in $\Rspace^4$ to intersection numbers between two 2-chains in $\Rspace^4$. This is the main point of the seemingly technical proof of Lemma 1. This lemma also is not difficult to prove using only properties of intersection and linking numbers. The basic idea of proving non-embeddability by presenting disjoint disks that bound linking cycles can also be found in \cite{Sko03}, Example 2. This idea was rediscovered by the author. \footnote{
The author is grateful to A. Skopenkov for notifying him that translating realizability in $\Rspace^4$ to linking numbers in the link-complex seems to have been done for the first time in \cite{Sko03}. The author reached this idea in his efforts to generalize the restriction on link-complexes of 2-complexes in 3-space, as explained in the introduction. 
Moreover, A. Skopenkov has notified the author that he exposed the ideas of [18] in a talk on [23] in August 2015 at IST Austria. The author attended to this talk.   However, the idea of the proof of Lemma 1, is from November 2014, and before the authors PhD defense.}
	
\vspace{0.75cm}

The above is true even for $d=1$. In a planar graph, triple intersections of link-complexes (subsets of vertices)
have at most three vertices each. This is because four points on a real-line always allow a link between two 0-dimensional spheres. Obviously,
in this case the graph would contain a $K_{3,3}$ otherwise, and hence would be non-planar. Of course, this bound is not tight since we know that
the triple intersections have at most two points. However, this example shows the application of the results of this article to the case $d=1$. This theorem is proved using only elementary properties of intersection homomorphism and linking numbers in $\Rspace^{2d}$.

In the case $d=2$ a stronger fact is true. A triple intersection of links is not only linklessly embeddable but actually is a planar graph\footnote{This improvement was noted by Uli Wagner.}.
The proof of planarity in the case $d=2$ uses the characterization of planar graphs by forbidden minors, and so is not directly generalized to higher dimensions, see the last section for a discussion about planarity.

These observations lead us to derive a new upper bound on the number of  $d$-simplices of embeddable $d$-complexes with $n$ vertices. For a simplicial complex $K$,
let $f_i(K)$ be the number of $i$-dimensional simplices of $K$. Denote by $f_d(n)$ the maximum number of $d$-simplices of an continuously\footnote{Our Theorem 1 is proved only for PL embeddings, however, Theorem 2 and its analog for linkless embeddings are proved for continuous embeddings. Since only these latter theorems are used to derive upper bounds, we define $f_d(n)$ with respect to continuous maps.} embeddable $d$-complex with $n$ vertices in $\Rspace^{2d}$.
The problem of determining or bounding $f_d(n)$ is a major open problem. For the case of (boundaries of simplicial) convex polytopes, by the famous Upper Bound Theorem, the $f$-vector
is always bounded above by the $f$-vector of the cyclic polytope, see \cite{McM70} and \cite{ZiegP}. This result has been strengthened to include all complexes homeomorphic to the boundary of a convex polytope, see \cite{Sta96}.
We note that deriving asymptotic tight bounds for all these cases is much easier by using the vanishing of the Betti numbers and the Euler relation. For instance, in the case $d=2$, the Euler relation
states that $\beta_0(K) - \beta_1(K)+\beta_2(K) = f_0(K) - f_1(K) + f_2(K)$ in a simplicial complex. Hence $f_2(K)$ asymptotically is dominated by $f_1(K) + \beta_2(K)$.

It is conjectured by many that the same (at least) asymptotic bounds that are true for $d$-simplices in the Upper Bound Theorem are also true for $f_d(n)$, this means
$f_d(n) = \Theta(n^{d})$. However, the best bound in the literature is $f_d(n) = O(n^{d+1 - 1/3^d})$, and this bound was the best known for any $d>1$. This fact is proved by forbidding some
non-embeddable subcomplexes. This bound is first mentioned in \cite{Dey93}, see \cite{Gun09} or \cite{Erd64,Wag11} for a proof, and \cite{Dey93} for an application. In this paper, in Theorem \ref{theorem:main}, we improve this bound to $f_d(n) = O(n^{d+1 - {1}/3^{(d-1)}})$. We prove the new bound in general dimension using further non-embeddability results of Gr\"unbaum \cite{Gru69}. Theorem \ref{theorem:basic} and its proof which is independent of the dimension give us a necessity condition for embeddability of $d$-complexes in $2d$-space. However, Theorem \ref{theorem:basic}, or the stronger planarity result mentioned above, help directly improve the bound on the number of simplices only for $d=2$.

We also show that a similar asymptotic bound (with different constant) is true for the number of $d$-simplices in $d$-complexes that can be linklessly embedded in $\Rspace^{2d+1}$. This is proved using the results in \cite{Sko03}. There exist also bounds on
these complexes by forbidding ``bad'' subcomplexes, see \cite{Tuf13,Sko03} for instances of small non-linklessly embeddable complexes.

It is shown in \cite{Wag11} that a (suitably defined) random $d$-complex embeddable in $\Rspace^{2d}$ has asymptotically almost surely $f_d(K) < C f_{d-1}(K)$ for some constant $C$. So the conjecture is true for almost
all complexes. The general belief is that the current upper bounds are far from the truth. Nevertheless, we think this paper leads to a better understanding of embeddable complexes.

It is conjectured that for $2$-complexes that are embeddable in $\Rspace^4$ one has $f_2(K)\leq 4 f_1(K)$. Inequalities of this type where considered by Gr\"unbaum \cite{Gru70}. This inequality is called Gr\"unbaum conjecture, also sometimes Kalai's conjecture. We believe this paper is a step towards proving this conjecture.

\paragraph{organization of the paper}
We presented an overview of the results in Section \ref{section:intro} above. In Section \ref{section:background}, we
briefly explain the necessary background material and definitions. This section serves also to set up our notation.
In Section \ref{section:main} we state formally the lemma on linklessness of the triple intersection of link-complexes and prove it. Next, we prove the stronger fact for $d=2$ on planarity of the triple intersection of link-graphs. Then, in Section \ref{ssection:bound} we derive the bounds on the number of simplices. The argument makes use of a combinatorial lemma which is proved in \ref{ssection:comb}.

\section{Basic Concepts} \label{section:background}
In this section we recall some definitions and briefly review some preliminary facts used later in the paper. By a simplicial complex $K$ we mean an abstract complex, i.e., a set system over a finite set of \emph{vertices} satisfying the usual condition that any subset of an element of $K$ is also in $K$. Let $\#\sigma$ denote the number of elements of the set $\sigma$.  If $\sigma \in K$, its \emph{realization} $|\sigma|$ is a simplex in a euclidean space that has $\#\sigma$ points in general position as vertices. The dimension of $\sigma$ is the dimension of its realization, i.e., $\#\sigma-1$. The dimension of $K$ is the largest of the dimensions of its simplices. A realization of $K$ in a euclidean space is defined as follows. For each vertex choose a point of the space with the condition that all the simplices of $K$ are simultaneously realized, and moreover, the realizations of disjoint simplices are disjoint. The subset of the euclidean space which is the union of realizations of simplices is a \emph{realization} of $K$ in the euclidean space. In fact, a realization always exists, and, with the induced topology from the euclidean spaces these realizations are homeomorphic. Hence, there is a canonical topological space defined for $K$ which is called the \textit{underlying space} of $K$ and denoted by $|K|$. We write $V(K) = \{ v_1, \ldots, v_n \}$ for the set of vertices. The \emph{$k$-skeleton} of $K$, i.e., the subcomplex made of simplices of dimension up to $k$, is denoted by $K_k$. We also assume the empty element of $K$ has dimension $-1$.

\paragraph{stars and links}

The \emph{star} of a simplex $\sigma \in K$ is the set $st(\sigma) = \{ \tau \in K, \sigma \subset \tau \}$. The \emph{closed star} of $\sigma$, $St(\sigma)$, is the smallest subcomplex of $K$ containing $st(\sigma)$. The complex $St(x) - st(x)$ is called the \emph{link-complex} of $\sigma$ and denoted $L(\sigma)$. The closed star is the join of $L(\sigma)$ with $\sigma$.

We work with link-complexes of vertices only. The stars of vertices cover $K$ such that each $k$-simplex, $k>0$, is covered $(k+1)$-times. Also, any $k$-simplex appears in as many link-complexes (of vertices), as the number of its incident $(k+1)$-simplices, or its \emph{degree}. It follows that the link-complexes of vertices cover all simplices of degree at least~$1$.

Let $f_k = f_k(K)$ denote the number of $k$-simplices in $K$, $k= -1, 0, \ldots$.
A $k$-simplex $\sigma \in K$ is determined by giving one of its vertices $v$ and the $(k-1)$-simplex of $L(v)$ whose join with $v$ is $\sigma$. Each
$k$-simplex is determined in $k+1$ different ways. Therefore, in general, the numbers $f_k$, $k \geq 0$, satisfy

\begin{equation} \label{formula:faces}
(k+1) f_k = \sum_{i=1}^{n} f_{k-1}(L(v_i)).
\end{equation}

\paragraph{notions of embeddings}
There exist various types of embedding into euclidean spaces. A continuous injective map is the most general notion. And the narrowest concept
for our purposes is the \textit{simplex-wise linear} embedding. This is the same as realizing the complex in the required euclidean space.
A \textit{piece-wise linear} (PL) embedding is one that is a simplex-wise linear embedding after finitely many (barycentric) subdivisions.
Since a closed simplex is compact, a continuous map can be approximated by a PL map such that the images of two (vertex) disjoint simplices
are disjoint. Such a PL map is called an \textit{almost embedding}. It would be interesting to know if the linklessness condition of Theorem \ref{theorem:basic} can be extended
to almost embeddings.

\paragraph {embeddings of link-complexes}
The concepts and notations in this paragraph are used throughout the paper and are important for us.
Let $K$ be a $d$-complex with a PL map $\iota: |K| \rightarrow \Rspace^{2d}$. Put a ball of small radius $\epsilon>0$ around the point $\iota(v_i)$, denoted by $B(v_i)$, and write $S(v_i)$ for its boundary $(2d-1)$-sphere. If we choose $\epsilon$ small enough then $S(v_i) \cap \iota(|K|)$ defines an embedding of the link-complex of $v_i$ into the sphere $S(v_i)$ and hence into $\Rspace^{2d-1}$. All the embeddings that are achieved in this way on spheres of different (small) radii are isotopic to each other. We refer to such an embedding when we say \emph{the embedding of the link-complex of} $v_i$, $i=1,\ldots,n$, and denote it by $\iota_{v_i} : |L(v_i)| \rightarrow S(v_i)$.

\paragraph{chains in spaces}
We will need familiarity with the very basic notions of chain complexes. In this paragraph we merely fix notation and refer to \cite{Hat01,Mun84,SeTh34} or any textbook of algebraic topology for complete definitions.

A singular $k$-dimensional simplex in a space $X$ is a (continuous) map $\sigma : |\Delta^k| \rightarrow X$, where $\Delta^k$ is a standard oriented $k$-simplex. The $k$-th singular chain group of $X$, $C_k(X)$ is a free abelian group generated by all singular simplices, where $-\sigma$ is $\sigma$ with the oppositely oriented domain simplex. The elements $c \in C_k(X)$ are called singular $k$-chains and they can be written as finite sums $c = \sum_i n_i \sigma_i$ where the $n_i$ are integers and the $\sigma_i$ are singular simplices.
When $X$ is the underlying space of a simplicial complex, one can in the above definition replace a singular simplex by a fixed linear homeomorphism of the standard oriented simplex onto a target simplex of the same dimension. Those homeomorphisms that preserve the orientation are declared equal and denoted by $\sigma$, and those which reverse the orientation are also declared equal and denoted by $-\sigma$, where $\sigma$ is the target simplex. The resulting chains are called simplicial chains. Hence a simplicial chain is a finite summation of oriented simplices with integer coefficients, and it can be viewed as a subset of all the singular chains closed under group operations. We also set $C(X) = \bigoplus_k C_k(X)$.

For each $k$, there exists a homomorphism $\partial_k : C_k(X) \rightarrow C_{k-1}(X)$ called the \textit{boundary homomorphism}. We refer to basic algebraic topology textbooks for the definitions. Intuitively, in the case of simplicial chains, $\partial_k$ assigns
to each generator the sum of its boundary codimension 1 simplices with proper signs. In the singular case, the boundary homomorphism assigns to a singular simplex (that is, a generator) a combination of restrictions of the singular simplex to the boundary simplices; the singular boundary is the image of the simplicial boundary of the domain.
A chain is said to \emph{bound} if there exists a higher dimensional chain that maps to it by $\partial_*$. A chain is a \emph{cycle} if its boundary is zero. We denote the group of $k$-cycles by $Z_k(X)$. We say two chains $c_1$ and $c_2$ are \emph{disjoint} if their images are disjoint.
A map $f$ between spaces induces homomorphisms on chain complexes which we denote by $f_\sharp$.
By $|c|$ we denote the image of the singular chain $c$, or its \emph{carrier}. Note that, if $c$ is a $k$-dimensional simplicial chain of a simplicial complex, then $|c|$ is the union of $k$-simplices that have non-zero coefficient in the unique presentation of $c$ in the basis formed by all of the $k$-simplices.

\paragraph{intersection homomorphism and linking numbers}
We make use of some elementary facts about intersection and linking numbers of chains in a euclidean space $\Rspace^d$ or in the sphere $S^d$. Here we present an overview on these important tools from algebraic topology. For proofs of these properties we refer to \cite{SeTh34,Lef56}. In $\Rspace^d$, for some integer $d>0$, the \emph{intersection number} of two singular chains $c_p\in C_p(\Rspace^d), c_{d-p} \in C_{d-p}(\Rspace^d)$ is an integer defined whenever $\partial c_p$ is disjoint from $c_{d-p}$, and $\partial c_{d-p}$ is disjoint from $c_p$, and moreover, the maps are ``nice'', see \cite{SeTh34}. It is enough for our purposes to restrict to pairs of singular chains that intersect finitely many times and transversely at each intersection point. Intuitively, the intersection number, $\inter(c_p,c_{d-p})$, counts the number of transverse intersections with proper signs. We next present a more formal introduction to the intersection numbers of chains in manifolds.

Let $M$ be an orientable closed triangulated $d$-manifold. Then, it is well-known that there exist dual cellular subdivisions \footnote{A cellular subdivision is a subdivision into polyhedral cells, instead of simplices in a triangulation. Cellular chains are defined analogously to simplicial chains.} for $M$. Let $T_1$ and $T_2$ be cellular subdivisions dual to each other. Orient the $d$-dimensional cells of $T_1$ and $T_2$ coherently, that is, so that the induced orientations on each $(d-1)$-cell are opposite. Since the cellular subdivisions are dual to each other, for any $k$-cell of $T_1$ there exists exactly one $(d-k)$-cell of $T_2$ with non-empty intersection, for $k=0, \ldots,d$. And that intersection is a single point and the intersection is transversal, meaning the two intersecting cells near the intersection point span a $d$-dimensional space.
A dual cellular subdivison to any triangulation can be obtained using a barycenteric subdivision of the triangulation. Then, dual to a $(d-k)$-simplex of the triangulation, is a $k$-cell which is the union of all $k$-simplices of the barycentric subdivision incident on the vertex added on the $(d-k)$-cell -- and used for its subdivision -- that are not inside the $(d-k)$-simplex, for $k=1, \ldots, d$. The dual of a $d$-simplex is the vertex of the subdivision inside it.

Assume now $\epsilon_k \sigma_k$ be an oriented $k$-cell of $T_1$, and, $\epsilon_{d-k} \tau_{d-k}$ be an oriented $(d-k)$-cell of $T_2$ where $\sigma$ is dual to $\tau$, and the $\epsilon_i$'s encode the respective orientations. And let $\epsilon$ encode the orientation of the manifold. That is, changing the orientation of the manifold multiplies $\epsilon$ by $-1$. Note that $\epsilon_k, \epsilon_{d-k}, \epsilon \in \{-1,1\}$.
Then, define the intersection number as
$$\inter(\epsilon_k \sigma_k, \epsilon_{d-k} \tau_{d-k}) = \epsilon_k\epsilon_{d-k}\epsilon.$$

For non-dual pairs of oriented simplices the intersection number is defined to be zero. The intersection number then extends bilinearly to all the cellular chains of $T_1$, $T_2$.

For two singular chains which satisfy the ``niceness" condition mentioned above, it is possible to approximate the two chains by chains in some two (very fine) dual cell subdivisions, and use the above intersection number definition. The standard theory, see e.g. \cite{SeTh34}, shows that the number so obtained is independent of the approximation.

The intersection number is bilinear as long as the terms on both sides are defined,
$$\inter(c_p+c'_p,c_{d-p}) = \inter(c_p,c_{d-p}) + \inter(c'_p,c_{d-p}).$$
Thus the intersection number defines a homomorphism $Z_p(\Rspace^d) \times Z_{d-p}(\Rspace^d) \rightarrow \Zring$. It also passes to homology groups,
that is if $c_p - c'_p = \partial d$ for a chain $d$ then $\inter(c_p,c_{d-p}) = \inter(c'_p, c_{d-p})$ whenever both terms are defined. We will not need this fact though. The crucial fact we use is that in $\Rspace^d$ any two cycles of complementary dimensions $(>0)$ have intersection number $0$. Both of these claims above follow from the following general fact about the intersection numbers, which again is true when the terms are defined, that is, when $\partial c_p$ and $\partial c_{d-p+1}$ are disjoint,

\begin{equation}
\inter(c_p, \partial c_{d-p+1}) = (-1)^{p} \inter(\partial c_p, c_{d-p+1}).
\label{f:b}
\end{equation}

We next define the linking number of two (null-homologous) disjoint cycles $z_p,z_{d-p-1}$ in $\Rspace^d$. Let $c$ be such that $\partial c = z_{d-p-1}$, such a chain $c$ always exists since the homology groups of $\Rspace^d$ are trivial. Then
$$ \link(z_p,z_{d-p-1}) = \inter (z_p,c)$$
is the \emph{linking number} of $z_p$ and $z_{d-p-1}$. It follows from (\ref{f:b}) that the linking number is independent of the choice of $c$ and hence is well-defined. We also note that the linking number in general changes within a homology class.

\paragraph{linklessly embeddable complexes}
Let $L$ be a simplicial $d$-complex and $\iota: |L| \rightarrow \Rspace^{2d+1}$ an embedding, and denote by $\iota_\sharp$ the induced map on chain groups. We call the embedding $\iota: |L| \rightarrow \Rspace^{2d+1}$ \emph{linkless\footnote{The name ``linking'' in this sense is related to the name ``link-complex'' since in a manifold the link-complex is a sphere linked with the original simplex.}}, if for any two disjoint simplicial $d$-cycles $c_1,c_2 \in C_d(L)$ their images $\iota_\sharp(c_1), \iota_\sharp(c_2)$ have linking number zero.  A simplicial $d$-complex is \emph{linklessly embeddable} if there exists a linkless embebdding of it into $\Rspace^{2d+1}$. We remark that there exist other definitions of linklessness of embeddings, but this definition is suitable for our application.

\paragraph{linklessly embeddable graphs}
The Conway-Gordon-Sachs theorem states that the graph $K_6$ is not linklessly embeddable into $\Rspace^3$. It follows that any graph which has a subdivision of $K_6$ as a subgraph is also not linklessly embeddable. It is an old and basic result in extremal graph theory, proved by Mader \cite{Mad68}, that a graph without a subdivision of $K_6$ as a subgraph satisfies $m \leq 4n$, where $n$ is the number of vertices and $m$ is the number of edges. This bound is tight and a graph with $4n + O(1)$ edges is just an apex graph, which is a planar graph together with a new vertex connected to every other vertex. On the other hand, there exists the Robertson-Saymour-Thomas characterization of linklessly embeddable graphs by forbidding the so called Petersen family of graphs as minors, \cite{RST95}.
Since $K_6$ is one of them, the set of linklessly embeddable graphs is contained in the set of $K_6$-minor-free graphs, and bounds for sizes of arbitrary $K_t$-minor-free graphs are also known, \cite{Tho01}.
Moreover, tight bounds on sizes of graphs that do not have a $K_t$ as topological subgraph are also known, \cite{BoTh98,KoSz96}.

\section{Links in Link-complexes} \label{section:main}
In this section we prove our main theorem on the possible link-complexes of vertices of a $d$-complex PL embedded into the euclidean $2d$-space. This theorem gives an obstruction for embeddability of complexes in euclidean spaces based on the complexity of the intersections of link-complexes of vertices.

Let $c \in C(K)$ be a simplicial chain. Whenever $c$ is defined in a link-complex $L(v)$ (i.e., $|c| \subset |L(v)| \subset |K|$) then we have a singular chain ${\iota_v}_\sharp(c)$, which is an embedding of the carrier of $c$ into $S(v)$. Recall that for every vertex $v$, $S(v)$ is a $(2d-1)$-sphere around $v$, which bounds a ball $B(v)$ centered at $v$. We assume the balls are so small that the embedding inside the preimage of the ball $B(v)$ (and of a slightly larger ball) is linear and the image inside the ball consists of a single connected component.

\begin{Lemma}\label{lemma:basic}
Let $\iota: |K| \rightarrow \Rspace^{2d}$ be a PL embedding of the $d$-complex $K$. Let $c_1,c_2 \in Z_{d-1}(K)$ be disjoint simplicial cycles. Give the spheres $S(v_i)$ orientations induced from that of $\Rspace^{2d}$. Assume there is a vertex $v$ such that $\link(\iota_{v\sharp}(c_1), \iota_{v\sharp}(c_2)) = \lambda \neq 0$ in $S(v)$, then $\link(\iota_{w\sharp}(c_1), \iota_{w\sharp}(c_2)) = \pm \lambda$ in $S(w)$, for any vertex $w$ for which $\iota_{w\sharp}(c_1)$ and $\iota_{w\sharp}(c_2)$ are defined. Moreover, if one such  $w \neq v$ exists then none of the chains $c_1, c_2$ can appear in a third link-complex.
\end{Lemma}

\paragraph{Remark.} Before presenting the proof we mention some points regarding it. 1) In the proof we consider simplicial chains as a special case of singular chains, that is, the simplicial chain complex is considered to be a sub-chain complex of singular chain complex. This is justified since any simplicial chain complex determines in a canonical way a singular chain complex. It is enough to observe this fact for the simplices. They are made into singular chains via the characteristic maps, see e.g. \cite{Hat01}, Chapter 2, see also Section \ref{section:background}.
2) For any vertex $v$, we regard $C(S(v))$ and $C(B(v))$ as sub-chain complexes of $C(\Rspace^{2d})$. This is obviously justified, since the same singular chain can be considered a chain in a larger space.

\begin{proof}
We refer to Figure \ref{fig:notation} for a schematic overview of the notation.
Let $c^u_{i} =\iota_{u\sharp}(c_{i})$, for a vertex $u$ and $i=1,2$.
By assumption, the linking number $\link(c^v_1, c^v_2)$ is defined and is not zero in $S(v)$.

\begin{figure}[h]
 \centering
 \includegraphics[scale=0.6]{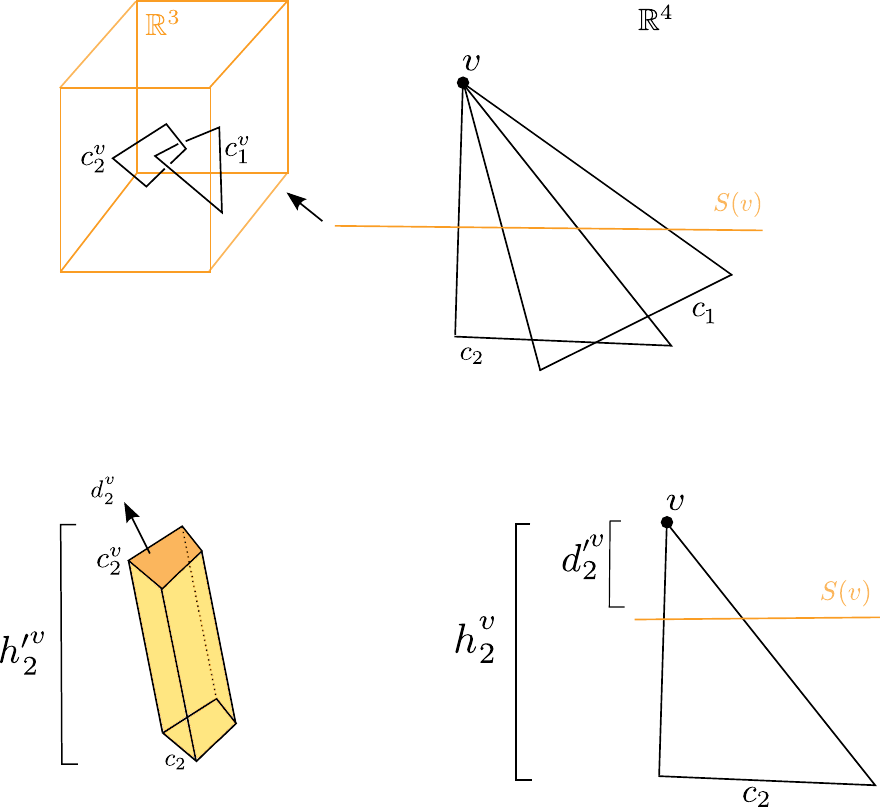}
 \caption{guide to the notation for $d=2$}
 \label{fig:notation}
\end{figure}

Let $d^v_1,d^v_2 \in C(S(v))$ be such that $\partial d^v_1 = c^v_1, \partial d^v_2 = c^v_2$. Then by definition we have
$$\lambda = \link(c^v_1, c^v_2) = \inter(c^v_1,d^v_2).$$
Here and in the following the ambient space in which a linking or an intersection number is computed is clear from the arguments of $\link(,)$ and $\inter(,)$.

Assume $c_1 =  \sum_{i=1}^{m} n_i \sigma_i$, were the $\sigma_i$ are oriented simplices and the $n_i$ are non-zero integers. If the oriented $d$-simplex $\sigma$ is defined by $v_0 \dotsc v_{d}$ then define $v \sigma$ to be the oriented $(d+1)$-simplex $v v_0 v_1 \dotsc v_{d}$, i.e., the cone over $\sigma$ oriented as indicated.
Let $h^v_1 \in C(\Rspace^{2d})$ be the chain $\iota_{\sharp}(v c_1)$, where $vc_1 = \sum_{i=1}^m n_i v\sigma_i$. Let $h^v_2$ be the same for $c_2$.  One easily computes that $\partial(vc_1) = \pm c_1$.
The sign is not changed if the convention is that the oriented boundary simplex $p_{d-1} \ldots p_0$ is considered positively oriented with respect to the oriented simplex $p_{d}p_{d-1} \ldots p_0$, we assume it is the case. It then follows that $\partial h^v_1 = \iota_\sharp(c_1)$.

Observe that the image of the complex inside $B(u)$ for any vertex $u$, is the linear cone over $S(u) \cap \iota(|K|)$ from $\iota(u)$.
Consider the singular chain $d'^v_2 \in C(B(v)) \subset C(\Rspace^{2d})$ whose image is the cone over $|c^v_2|$ from $\iota(v)$, $c^v_2 \in S(v)$, and with $\partial d'^v_2 = c^v_2$. This singular chain clearly exists and is easy to construct from $c^v_2$. Define $h'^v_2 = h^v_2-d'^v_2+d^v_2 \in C_d(\Rspace^{2d})$. We claim that the intersection numbers satisfy

\begin{equation} \inter(h^v_1, h'^v_2) = \pm \inter (c^v_1, d^v_2) \label{eq:intersection} \end{equation}

where the intersection number in the right-hand-side is computed in $S(v)$.
Assume there exists a vertex $w \neq v$, as in the statement of the lemma i.e. such that $c_1$ and $c_2$ appear in the link of $w$, with $\iota_w: |L(w)| \rightarrow S(w)$.
Assuming \eqref{eq:intersection} we argue as follows that $\link(c^w_1, c^w_2) = \pm \lambda$ in $S(w)$. Note that
this latter condition implies that $c^w_1$ and $c^w_2$ are defined.

To see this claim, we write
\begin{equation} \inter(h^w_1 - h^v_1, h'^w_2-h'^v_2) = \\
\inter(h^w_1, h'^w_2) - \inter(h^w_1, h'^v_2) - \inter(h^v_1,h'^w_2) + \inter(h^v_1, h'^v_2)\label{eq:cycleinter} .
\end{equation}

The two middle terms of the right-hand side are zero as follows. We have $$\inter(h^w_1, h'^v_2) = \inter(h^w_1, h^v_2) + \inter(h^w_1, - d'^v_2 + d^v_2).$$
The underlying spaces of the chains $h^w_1$ and $h^v_2$ are disjoint and hence these two chains have intersection number zero. The second term is also zero since the cycle $d^v_2 - d'^v_2$ is inside the ball bounded by $S(v)$ and is disjoint
from $h^w_1$ by the choice of the spheres. The third term of the right-hand side of \eqref{eq:cycleinter} can similarly be shown to be zero.

Since the two cycles $h^w_1 - h^v_1$ and $h'^w_2-h'^v_2$ must have intersection number $0$, it follows that $\inter(h^w_1, h'^w_2) = -\inter(h^v_1, h'^v_2)$. And from \eqref{eq:intersection} it must be
that $\inter(h^w_1, h'^w_2) = \pm \inter(c^w_1, d^w_2)= \pm \link(c_1^w, c_2^w)$, and similarly, $\inter(h^v_1, h'^v_2) = \link(c^v_1, c^v_2)$.
Therefore, $\link(c^w_1, c^w_2) = \pm \lambda$.

It remains to show \eqref{eq:intersection}. Since by assumption the map $\iota$ restricted to $\iota^{-1}(\hat{B}(v))$ is linear, where $\hat{B}(v)$ is a ball slightly larger than $B(v)$, any transverse intersection between $|c^v_1|$ and $|d^v_2|$ gives rise to a transverse intersection between $|h^v_1|$ and $|d^v_2| \subset |h'^v_2|$.
The rest of $|h'^v_2|$ is disjoint from $|h^v_1|$. Hence, the set of intersection points of $|h^v_1|$ and $|h'^v_2|$ is the same as the set of intersection points of $|c^v_1|$ and $|d^v_2|$ in $S(v)$.
 We thus only need to argue that the intersections have the same signs in $\Rspace^{2d}$ as in $S(v)$, or all of the signs are changed. The sign of an intersection depends on three orientations, the two of the intersecting chains around the intersection point and the orientation of the ambient space. The chains $d^v_2$ and $h'^v_2$ have the same orientations around intersection points. The orientations of the $S(v_i)$ are induced by that of  $\Rspace^{2d}$ and hence are fixed. We thus must show that the orientations of the simplices of $h^v_1$ are determined naturally around the intersection points by those of $c^v_1$. But this is the case since the orientations of the simplices of $h^v_1$ inside the ball $\hat{B}(v)$ are defined from those of $c^v_1$ by the natural coning process sending $\iota_v(\sigma_i)$ to $\iota(v) \iota_v(\sigma_i)$.

Next we prove the last part of the lemma.
From the above it follows that the $2$-cycle $s_1 = h^v_1-h^w_1$ has linking number $\pm \lambda$ with $c^v_2$. Since $c^v_2 - \iota_\sharp(c_2)$ bounds a chain which is disjoint from $s_1$, it follows that $\link(s_1,\iota_\sharp(c_2))= \pm \lambda$.
If a third vertex $u$ exists such that $c_2$ is in its link-complex, then $\iota(uc_2)$ would be a chain disjoint from $s_1$ and bounding $c_2$. This contradicts the fact that $\lambda \neq 0$. Symmetrically the argument works for $c_1$ as well.$\square$
\end{proof}

Theorem \ref{theorem:basic} is an immediate corollary of the above Lemma.

\paragraph{Remark}
We have presented our main lemma above in a setting that provides more information that is needed for Theorem \ref{theorem:basic}. The reader will realize that, for instance, for chains with $\Zring_2$ coefficients and $\Zring_2$ intersection numbers, the proof simplifies.

\subsection{Planarity for $d=2$}

\begin{Theorem}
Let $K$ be a $2$-complex embedded in $\Rspace^{4}$. Let $L$ be a 1-subcomplex that is the intersection of three link subcomplexes of $K$. Then, $L$ is a planar graph.
\end{Theorem}

This theorem can be derived easily from the following fact first proved by Gr\"{u}nbaum \cite{Gru69}, see also \cite{Umm73,Mat03}. Let $K^{d_1}_{2d_1+3}, K^{d_2}_{2d_2+3}, \ldots, K^{d_p}_{2d_p+3}$ be
$p$ complexes such that $K^{d_i}_{2d_i+3}$ is the complete $d_i$-complex on $2d_i+3$ vertices, and, $d=d_1+d_2+ \ldots+d_p+p-1$. Then
$$ K' = K^{d_1}_{2d_1+3} * K^{d_2}_{2d_2+3} * \dots * K^{d_p}_{2d_p+3}$$ is a $d$-complex not embeddable in $\Rspace^{2d}$. The complex $K'$ is also minimal in the sense that removing a single $d$-simplex makes it embeddable. In order to show $L$ planar, we must show that as a topological graph it cannot contain a topological $K_{3,3}$ or $K_5$. We have,
$$K_{3,3} = [3]*[3], K_5 = K^1_{2\times 1+3} $$
where $[3]$ is a complex consisting of three disjoint vertices. Thus, if $L$ contains a homeomorphic copy of a $K_{3,3}$ or $K_5$ then $K$ will contain a subcomplex
homeomorphic to $[3]*[3]*[3]$ or $K^1_5 * [3]$.
However, by the result mentioned above these complexes are not embeddable and the theorem is proved.

\section{Bounding the Number of  $d$-Simplices in $d$-Complexes in $\Rspace^{2d}$} \label{ssection:bound}
In this section we use the theorems which restrict the triple intersections of links of vertices to derive an upper bound
on the number of $d$-simplices in a $d$-complex embeddable in $\Rspace^{2d}$.

In the following, we consider first the cases $d=2$ and $d=3$ as examples of the general case and to introduce the intuition behind the proof.
Let $d=2$. For the purpose of bounding the number of triangles, we use the fact that triple intersections of link complexes are planar. It is well-known that a planar graph on $n$ vertices has at most $f(n) = 3n-6$ edges.
Thus, for any embedded 2-complex $K$ in $\Rspace^4$, any triple intersection of links of vertices has at most $f(n) = 3n-6$ edges. From (\ref{formula:faces}) and a basic combinatorial lemma proved in the next section, Lemma \ref{lemma:comb}, it follows that the total number of triangles satisfies $f_2 = O(n^{8/3})$.

A similar result with a slightly worse constant and only for PL embeddings can be obtained using Theorem \ref{theorem:basic}.
This is because the graph $K_6$ is not linklessly embeddable by the well known Conway-Gordon-Sachs theorem. Hence, any linklessly embeddable graph cannot contain a subdivision of $K_6$ and by the extremal results of Mader \cite{Mad68} it has at most $4n+O(1)$ edges.

Next let $d = 3$. From the non-embedd\-ability result of Gr\"{u}nbaum it follows that a triple intersection of link-complexes
cannot include a complex $F = K^1_5 * [3]$. This is because non-embeddability of $K^1_5*[3]*[3]$ (into $\Rspace^6$) can be ``read from the right'' to imply that: In an embeddable complex,
any triple intersection of link-complexes of vertices is such that any triple intersection of link-complexes of it does not include a homeomorphic image of $K^1_5$ as a subgraph. Thus, by the above discussion, any triple intersection of links is
a $2$-complex of size $O(n^{8/3})$. By Lemma \ref{lemma:comb} the total number of $d$-simplices is $O(n n^2 n^{8/9})=O(n^{3+8/9})$.

We now state the general result which is proved similarly.

\begin{Theorem}\label{theorem:main}
Let $f_d(n)$ be the maximum number of $d$ simplices in an $n$-vertex $d$-complex embeddable in $\Rspace^{2d}$, $d>0$. Then
\begin{equation}
 f_d(n) = O(n^{d+1 - \frac{1}{3^{(d-1)}}}).
\end{equation}
\end{Theorem}

\begin{proof}
For the case $d=1$ the above reduces to $f_1(n) = O(n)$ hence the classical bound on the number of edges of a planar graph. So we assume that $d>1$.
Let $\phi_d(n)$ be the maximum number of $d$-simplices in an $n$-vertex $d$-complex that does not have a subcomplex homeomorphic to $K_5 * [3] * \cdots * [3]$ with
$d-1$ factors of $[3]$. By the results of Gr\"{u}nbaum $f_d(n) \leq \phi_d(n)$ and we bound $\phi_d(n)$. For $d=2$, the condition that the complex does not contain
a homeomorphic copy of $K_5 * [3]$ implies that triple intersections of link-graphs are graphs that do not contain a subdivision of $K_5$. By the early results of Mader \cite{Mad68}
such graphs have $O(n)$ edges. Then, from Lemma \ref{lemma:comb} it follows that $\phi_2(n) = O(n^{8/3})$.

Now consider the case of general $d$. The condition that a $d$-complex does not have a subcomplex homeomorphic to $K_5 * [3] * \cdots * [3]$ with
$d-1$ factors of $[3]$ implies that triple intersections of link-complexes do not contain subcomplexes homeomorphic to $K_5 * [3] * \cdots * [3]$ with $d-2$ factors
of $[3]$. Hence, triple intersections of link complexes are $(d-1)$-complexes over at most $n-1$ vertices whose number of $(d-1)$-simplices is bounded by $\phi_{d-1}(n)$.
Therefore, we can apply Lemma \ref{lemma:comb} again and we obtain

\begin{equation}
\begin{split}
 \phi_d(n) \leq c n n^{\frac{2d}{3}} \phi_{d-1}^{\frac{1}{3}}(n) \leq c c' n n^{\frac{2d}{3}} n^{\frac{1}{3}(d - \frac{1}{3^{d-2}})} \\
 =cc' n^{d+1 - \frac{1}{3^{d-1}}}.
\end{split}
\end{equation}

In the above $c$ is a constant coming from the combinatorial lemma and $c'$ is the constant in the asymptotic bound for $\phi_{d-1}(n)$, so in general the constant in the notation depends on $d$ $\square$.
\end{proof}

\subsection{Size of Linklessly Embeddable $d$-Complexes in  $\Rspace^{2d+1}$}
Using the linklessness criterion of Theorem \ref{lemma:basic} one can prove in a way similar to the proof of Theorem \ref{theorem:main}, a stronger result than Theorem \ref{theorem:main} for continuous embeddings.
By \cite{Sko03}, Lemma 1, if a $d$-complex embeds linklessly in $\Rspace^{2d+1}$ then it cannot contain a subcomplex homeomorphic to
$[4] * \cdots *[4]$, $d+1$ factors. Using an inductive argument as in the proof of Theorem \ref{theorem:main}, it follows that the a similar asymptotic bound proved
above also applies to complexes that are linklessly embeddable in $\Rspace^{2d+1}$. In the argument, $K_5$ is replaced by $[4] * [4]$ and Lemma \ref{lemma:comb} is replaced by a corresponding lemma for 4-wise intersections. Hence,
for the case $d=2$ one needs to bound the number of edges in graphs with no subdivision of $K_{4,4}$ as a subgraph and for this purpose it is enough to consider sizes of graphs with no subdivision of $K_8$
as subgraph. This is because the class of graphs with no subdivision of $K_{4,4}$ as a subgraph is contained in the class of graphs with no subdivision of $K_8$ as a subgraph.

Note that here for $d>1$ the codimension is at least
three and a continuous embedding can always be approximated by a PL one \cite{Bry72}.

\begin{Theorem}
Let $g_d(n)$ be the maximum number of $d$-simplices of an $n$-vertex $d$-complex linklessly embeddable in $\Rspace^{2d+1}$, $d>0$. Then
\begin{equation}
 g_d(n) = O(n^{d+1 - \frac{1}{4^{(d-1)}}}).
\end{equation}
\end{Theorem}

\subsection{A Combinatorial Lemma} \label{ssection:comb}

This section gives the proof of a combinatorial lemma used in deriving the upper bounds. We produce it here for completeness and do not claim it is new. In this section, we denote the number of elements of a set $S$ by $|S|$.
Let $X = \{x_1, \ldots, x_n \}$ be a finite set and $S = \{S_1, S_2, \ldots, S_m \}$ a collection of subsets of $X$. We are interested in bounding the quantity $t(S) = \sum_{i=1}^{m} |S_i|$ that is a function of $n$ and $m$. The restriction on sets $S_i$ comes from their common intersections. Assume that each triple of distinct sets $S_i,S_j, S_k$ satisfies $$ |S_i \cap S_j \cap S_k| \leq f(n)$$ where $f(n)$ is a function of the total number of elements $n$.

\begin{Lemma} \label{lemma:comb}
For the set systems satisfying the above conditions and $m\geq 3$, $$t(S) = O (m n^{\frac{2}{3}} f(n)^{\frac{1}{3}}+n),$$ and the bound is best possible given only these conditions on the set systems.\footnote{The author thanks A. Kupavskii and A. Skopenkov for finding a counter-example to an earlier incorrect statement of the lemma.}
\end{Lemma}

\begin{proof}
Let $\kappa_i$ be the number of sets that the element $x_i$ belongs to. The totality of elements for which $\kappa_i <6$ contribute at most $6n$ to $t(S)$. Therefore, in the following we assume that $\kappa_i \geq 6$ for all $i$. To prove the lemma, we bound the quantity $\sum_{\{i,j,k \}} |S_i \cap S_j \cap S_k|$ in two ways. First, since there are ${m \choose 3}$ summands with each having at most $f(n)$ elements we have

\begin{equation*}
\sum_{\{i,j,k \}}  |S_i \cap S_j \cap S_k| \leq {m \choose 3} f(n).
\end{equation*}

Let the variable $Y_{lijk} = 1$ when $x_l$ is in $S_i \cap S_j \cap S_k$, $i<j<k$, and zero otherwise. Then, $\sum_{\{i,j,k \}} |S_i \cap S_j \cap S_k| = \sum_{l,i<j<k} Y_{lijk} = \sum_{l} {\sum_{i<j<k} {Y_{lijk}}}$. On the other hand, $\sum_{i<j<k} {Y_{lijk}}$ is the number of triples $(i,j,k), i<j<k,$ such that $x_l$ appears in $S_i \cap S_j \cap S_k$. This number is ${\kappa_l \choose 3}$. Then we have

\begin{equation} \label{eq:2}
 \sum_l {\kappa_l \choose 3} = \sum_{\{i,j,k \}}  |S_i \cap S_j \cap S_k|.
\end{equation}

By the H\"older inequality

$$\kappa_1+\kappa_2 + \ldots + \kappa_n \leq (\kappa_1^3 + \kappa_2^3 + \ldots + \kappa_n^3) ^{\frac{1}{3}} n^{ \frac{2}{3}}.$$

Writing $\kappa = (\sum \kappa_l) / n$, the above becomes
$ n \kappa^3 \leq \kappa_1^3 + \ldots + \kappa_n^3$.

We expand the left hand side of (\ref{eq:2}):
$$ \sum_l {\kappa_l \choose 3} = {\frac{1} {3!}}( \sum_l {\kappa_l^3} -3 \sum_l \kappa_l^2 + 2 \sum_l \kappa_l).$$

From above and the fact that the $\kappa_l$ are non-negative and at least 6 it follows that

\begin{equation}
\sum_l {\kappa_l \choose 3} \geq \frac{1}{2} \frac{1}{3!} n \kappa^3.
\end{equation}

Therefore,
$$ \kappa^3 = O( m^3 n^{-1} f(n)) $$
and since $n \kappa = t(S)$ we obtain
$$ t(S) = O(m n^{\frac{2}{3}} f(n)^{\frac{1}{3}}).$$

If we consider the set of all those set systems for which all the $\kappa_i$ have asymptotically the same order, then the H\"older inequality is tight
and it follows that for those sets $t(S) = \Theta(m n^{2/3} f(n)^{1/3})$. Thus in general, using only the conditions in the lemma this bound cannot be improved. $\square$
\end{proof}

We now discuss set systems that can possibly achieve the bounds of the lemma with parameters that are of interest to us,
i.e., $n = f_0^2, m = f_0, f(n) = n^{1/2} = f_0$. Assume $S_1, \ldots, S_m$ are subsets such that each triple intersection (of distinct sets) $S_i \cap S_j \cap S_k$ has
exactly $f$ elements. Then, we can form another dual system
as follows. Let $Y = \{ S_1, \ldots, S_m \}$ and define $T_i$ to be those sets $S_j$ that contain the element  $x_i$, for
$i=1, \ldots, n$. Then the sets $T_i$ satisfy the following conditions. Each $3$-set of elements of $Y$ appears in exactly $f$
sets $T_i$. This defines a Steiner system $S_f(3,K;m)$, where $K$ is the set containing sizes of sets $T_i$ and with $n$ sets.

Refer to \cite{Bethetal99} for an extensive account of these and other combinatorial designs. Let $t$ be a positive integer. A \textit{ Steiner system}
$S_\lambda(t,K;v)$ is a set system $(A, \{B_1, \ldots, B_b \})$ over a set $A = \{a_1, \ldots, a_v \}$ such that, each set (or block) has size from $K$.
Moreover, every $t$-subset of $A$ appears in exactly $\lambda$ blocks.

It is clear from the above that any $S_f(3,K;m)$ with $n$ blocks can be used to build $m$ subsets of an $n$-element set such that each triple intersection of them has exactly
$f$ elements. Also, ``approximate'' Steiner systems are enough for our purposes, however, we are not aware of explicit descriptions of Steiner systems $S_{\lambda}(3,K; m)$
with roughly $f^2_0$ blocks such that $m$ and $\lambda$ are also approximately $f_0$. See \cite{Bethetal99}, Appendix 5 for a table of known Steiner systems.

\paragraph{Remark}
If we are given a sequence of set systems achieving the upper bound in Lemma \ref{lemma:comb}, then we can build as follows a simplicial complex
whose triple intersections of links have at most $f(n)$ elements each and with $t(S)$ triangles. Let $f_0$ be the number of vertices. Take the set from the lemma with $m = f_0, n = f_0^2,
f(n) = f_0$ (if such set systems exist). Then, simply make a graph over $2f_0$ vertices, so that the elements of the set $X$ can be identified with the edges. Then, one introduces a vertex for each set $S_i$ and cones over the
corresponding set of edges. The resulting complex has the required properties, i.e, the triple intersections
of link-complexes have at most $f_0$ edges and the complex contains $\Theta(f_0^{8/3})$ triangles.

\section{Discussion}
We have shown certain restrictions on intersections of link-complexes of vertices in embeddable simplicial complexes.
There is an obvious direction to continue this research and that is to find more restrictions of the type introduced here.

Let's say a simplicial complex is \textit{$d$-planar} if it embeds in the euclidean $d$-space.
It is natural to strengthen the linklessness criterion to $(2d-2)$-planarity for all $d>1$. It could well be possible in addition
to the case $d=2$ shown above, in dimensions $2d-2$, where the embeddability is characterized by the van Kampen obstruction, i.e., when $2d-2 \neq 4$.
One shows that if a complex has van Kampen obstruction nonzero, then its join by three vertices also has nonzero van Kampen obstruction. Hence, non-embeddable
complexes in euclidean space of dimension $2d-2 \neq 4$ cannot be triple intersections of link-complexes of embeddable $d$-complexes in  $2d$-space. However, we are more interested to know if a proof exists that works for all dimensions
and hence does not use the characterization of embeddable $d$-complexes in $2d$-space, by forbidden minors or the van Kampen class.

It was shown that a triple intersection of link-complexes of vertices of a $2$-complex in $\Rspace^4$ is a planar graph. It is interesting to know if
an embedding of the triple intersection graph in a plane or a $2$-sphere can be obtained from the given embedding of the complex.

\paragraph{Acknowledgements}
The author is indebted to Herbert Edelsbrunner for bringing Lemma \ref{lemma:comb} to his notice. The author also thanks Tamal Dey and
Uli Wagner for helpful discussions about the problem of this paper.

\bibliographystyle{spmpsci}      


\end{document}